\documentclass[11pt, a4paper]{article}
\usepackage{amsmath, amssymb, amsfonts, amsthm}
\usepackage{color}
\usepackage{bm}
\usepackage{subfigure}
\usepackage{graphicx}
\usepackage{mathabx}
\usepackage{multirow}
\usepackage{setspace}
\usepackage{algorithmic, epsf}
\usepackage{graphicx, color}
\usepackage{paralist}
\usepackage{float}
\usepackage{algorithm}
\usepackage{lineno}
\usepackage{verbatim}

\usepackage{fullpage}

\title{Small Model $2$-Complexes in $4$-Space and Applications\footnote{This resreach was done while the author was at IST Austria, and supported by TOPOSYS project FP7-ICT-318493-STREP.}}

\author{Salman Parsa\footnote{\'Ecole normale sup\'erieure, Paris, France. \texttt{parsa@di.ens.fr}}}
\date{}

\begin{document}

\maketitle
\newcommand{\Fspace}        {\mathbb F}
\newcommand{\Rspace}        {\mathbb R}
\newcommand{\Zring}        {\mathbb Z}
\newtheorem{observation}{Observation}
\newtheorem{lemma}{Lemma}
\newtheorem{claim}{Claim}
\newtheorem{Theorem}{Theorem}
\newtheorem{theorem}{Theorem}
\newtheorem{corollary}{Corollary}

\begin{abstract}
We consider computational complexity of problems related to the fundamental group and the first homology group of (embeddable) $2$-complexes. We show, as
an extension of an earlier work, that computing first homology of $2$-complexes is equivalent in computational complexity to matrix diagonalization. That is, the usual procedures
for computing homology cannot be improved other than by matrix methods. This is true even if the complex is in the euclidean $4$-space. For this purpose, we use $2$-complexes built in a standard way from group presentations, called model $2$-complexes. Model complexes have fundamental group isomorphic with the group defined by the presentation. 
We show that there are model complexes of size in the order of the bit-complexity of the presentation that can be realized linearly in $4$-space. We further derive some applications of this result regarding embeddability problems in the euclidean $4$-space.
\end{abstract}

\section{Introduction}
We study the fundamental groups of simplicial $2$-complexes which are
geometrically realized in the $4$-dimensional euclidean space. We approach this topic from a computational complexity point of view. It is known that a $2$-complex in the $3$-space cannot have torsion in the first and second homology groups. Thus, this embeddability property restricts the first homology group (hence also the fundamental group). The embeddability into $\Rspace^3$ also makes the computation of the homology groups and Betti numbers more efficient using Alexander duality \cite{DeEd95}. It is shown in \cite{EdPa14} that, from a computational point of view, no restriction applies to the first $\Zring_2$-homology of $2$-complexes embedded in $\Rspace^4$. This was done by reducing matrix rank computation to computing Betti numbers of $2$-complexes in $\Rspace^4$. For such a reduction in complexities of algorithms, one needs embedded complexes with given homology groups 
that moreover can be constructed efficiently, in almost linear time, from the given input. The input in the case of \cite{EdPa14} was a 0-1 matrix $M$, which was considered as a system of equations $Mx=0$, where $x$ is a 
vector of variables. Then an embedded simplicial complex was constructed in linear time such that its second $\Zring_2$-homology group is isomorphic with the null-space of $M$.
Now a generalization of this reduction of rank computation to homology computation of complexes in $\Rspace^4$ is to consider $\Zring$ coefficients and integer matrices, or 
in a more general setting, one can study what embeddability implies regarding the fundamental group. Such a generalization was the purpose of this work.

We are given a finite presentation $P = <g_1, \ldots, g_n ; r_1, \ldots, r_m>$ of a group, and, our goal is to construct a complex 
in $\Rspace^4$ whose fundamental group is isomorphic with the given group. In order to be useful for the reduction of algorithmic problems, the
embedding needs to be constructive. Moreover, we need the whole construction take an almost linear time. This is because we want to argue that, 
in terms of computational complexity, for problems about the fundamental group or the first homology group, we gain nothing when the $2$-complex is given as realized in $\Rspace^4$, compared to
an abstract $2$-complex.

We broadly define a \textit{model simplicial complex} to be one that is constructed in a standard way (i.e., by a fix algorithm) from a presentation of a group, and whose fundamental group is isomorphic with the presented group. This is in contrast to the common definition, for instance in \cite{HAMe93}, and more suited to our purposes. 
In the literature, there exist proofs of the fact that model $2$-complexes embed in $\Rspace^4$. The first mention of this result is in a statement\footnote{ No proof is given in \cite{Cur62} and Stallings also refers this case to Curtis \cite{Cur62}.} by Curtis \cite{Cur62}. Later Stallings generalized this fact and showed that any homotopy type of $n$-complexes
is embeddable in $\Rspace^{2n}$ \cite{Sta65}. See \cite{DrRe93} for a simple proof of Stalling's theorem. There are also results in the literature on embeddings of $2$-complexes with certain properties for the complementary space. In \cite{Huc90}, together with giving another proof of embeddability of model complexes, it is shown that acyclic $2$-complexes can be embedded with contractible complement, see also references in it. For a survey on possible thickenings of $2$-complexes in various dimensions see \cite{HAMe93}. 

In this article, we give an explicit linear embedding of \textit{a} model complex which is small, as is needed in the reduction arguments and to be made precise below, see Theorem \ref{theorem:bit}.
The embedding results are formalized and proved in Theorems \ref{theorem:main} and \ref{theorem:bit}. Our proof was derived independently from the existing literature.
We note that although our argument has turned out to be similar in nature to \cite{DrRe93}, the approach we use needs a modification of the given presentation (see the beginning of Section \ref{section:proofs}), and, we build the complex for the modified presentation. The embedding of the model complex is then linear. In addition, we mention some interesting corollaries of the main theorem, which are partly well-known, as examples of possible applications of such simple embeddability arguments, namely, Corollary \ref{cor:manifold} and Theorem \ref{theorem:htype}. We also generalize a simple theorem on undecidability of a relative embedding problem for dimensions $d \geq 4$, proved in \cite{FrKr14} for $d \geq 5$.

It is not difficult to compute the complexity
of embedding procedures presented in some of the existing literature, and they are at best linear in the unary size of the presentation (as defined in Section \ref{section:theorems}), and, not ``small'' in our sense. For us, \textit{small} means that the complex has number of simplices proportional to the bit-complexity
of the presentation, i.e., the number of bits required to encode the presentation. 
Certain other interests also exist in having small model complexes in euclidean $4$-space, see for instance \cite{LuZi08}.

One implication of the embeddability of model complexes in $\Rspace^4$ is that an algorithm for deciding whether a given $2$-complex (PL) embeds in $\Rspace^4$ does not decide any property of the fundamental group. This has been also one of the motivations for the present work.

\paragraph*{Complexity of Integral Homology Computation (for $2$-complexes in $\Rspace^4$)}
In \cite{EdPa14} the complexity of computing $\Zring_2$-Betti numbers was considered. A version of the main result of this paper can be expressed as generalizing the result of \cite{EdPa14} to integral homology computation.
It reads as follows. For any given abelian group (as a presentation), one can build
in almost linear time\footnote{up to a logarithmic factor, see definition of binary size of a presentation.} in \textit{size of the input} a $2$-complex in $\Rspace^4$ whose first integral homology group is isomorphic with the given group. Therefore, computing a decomposition
of an abelian group from its presentation reduces to computing first homology group of a complex in $\Rspace^4$. Where ``computing the homology group'' means
as usual expressing the homology group as a direct sum of cyclic groups from the boundary matrices. Note that for this result one needs the reduction to have complexity roughly bounded
by the bit-complexity of the input, hence the need for small model complexes.

The presentation of the abelian group can be thought of as a system of linear equations over the integers $Mx=0$. The matrix $M$ defines a linear transformation from a direct sum of $\Zring$ summands to again a direct
sum of $\Zring$ summands; $M:X \rightarrow X$. The group presented is $X/ \ker(M)$. The results proved here then shows that the group $X/ \ker(M)$ is the homology group of a $2$-complex
in $\Rspace^4$, which can be constructed in time almost linear in the bit-complexity of $M$. Observe that presentation defined by $M$ does not include relations for commutativity of the generators,
however, these are only necessary if we want to have a complex whose fundamental group is isomorphic with $X/ \ker(M)$.

On the other hand, since the homology of a $2$-complex is computed by diagonalization of a constant number of integer matrices, the above considerations imply that computing the homology of embedded complexes in $\Rspace^4$ and integer matrix diagonalization are equivalent in terms of computational complexity.

\paragraph*{}
The rest of the paper is organized mainly into two sections. In the next section we state the formal definitions and theorems and
the second section includes the proofs.

\section{Theorems}\label{section:theorems}
Let $P = \langle g_1, g_2, \ldots, g_n; r_1, \ldots, r_m \rangle$ be a presentation of the group $G = G(P)$ with generators $g_i$ and relations $r_j$. By definition, the
group $G$ is the quotient group $\langle g_1, \ldots,g_n \rangle  / (r_1, \ldots, r_m)$, where $ \langle g_1,\ldots,g_n \rangle $ is the free group generated by the $g_i$, and, $(r_1, \ldots, r_m)$
is the smallest normal subgroup containing the elements represented by the words $r_j$. To prevent some complications, we assume that the given presentation always has at least $n$ relations, hence,
free groups could be presented with $n$ copies of the symbol denoting the empty word for the $r_i$.

Given a presentation $P$ there is a well-known procedure for building a $2$-dimensional CW complex $K = K(P)$ whose fundamental group $\pi_1(K,x_0)$, with basepoint $x_0 \in |K|$, is isomorphic with the group $G$. One first creates a wedge $W$ of $n$ directed circles $g_1, \ldots, g_n$. The basepoint is the wedge point. For each element $r_j$, $j=1,\ldots,m$, one takes a $2$-disk $D_j$ and attaches its boundary to the complex along the loop $r_j$ based at the wedge point. That is, by a map $f: S^1 = \partial D_j \rightarrow W$ such that $f$ represents the element represented by the word $r_j$ in $\pi_1(W)$ based at the wedge point. It follows from the Van Kampen-Seifert theorem that the fundamental group of the resulting complex is isomorphic with $G$. See \cite{Si93} for details of this construction and for an extensive account of algebraic topology of $2$-dimensional complexes.

The construction above can be done such that $K(P)$ is a simplicial complex. We explain a specialized construction of a simplicial complex adapted to our needs. This procedure defines our notion of the model complexes. We first take a wedge of simplicial circles, say each with $c$ edges. For each relation $r_j$, $j=1,\ldots,m$, we do the following. Let $r_j = g'_1 \dotsb g'_{k_j}$, where for $l = 1, \ldots,k_j$, $g'_l=g_i$ or $g'_l = g^{-1}_i$ for some $i$. We attach an annulus, denoted $A_j$, consisting of $c k_j$ rectangles to $W$ along its \emph{lower boundary}, where the \emph{lower edge} of each rectangle is identified with an edge of $W$. The lower boundary of the annulus is attached such that the attaching map represents $r_j$. In the next step, a disk is attached to the \emph{upper boundary} of the annulus by a degree 1 map to finish attaching the disk for $r_j$. We note that the size of the resulting complex is $O(n+\sum_{j=1}^{m} k_j)$. We call the number $s(P) = n+\sum_{j=1}^{m} k_j$ the \emph{unary size}\footnote{Note that to encode
the presentation one needs an additional $\log{n}$ multiplicative factor for addressing the generators, for simplicity we suppress this factor throughout.} of the presentation $P$.

We say that a simplicial complex is \emph{realized} in $\Rspace^4$ if it is
simplex-wise linearly embedded. We now state the main theorem.

\begin{theorem}\label{theorem:main}
There is an algorithm that, given any presentation $P$, constructs the model simplicial complex $K = K(P)$ realized in $\Rspace^4$. Moreover, the run-time of the algorithm (and hence size of $K$) is of the order of the unary size of $P$. 
\end{theorem}

We mention an interesting corollary of this theorem before proceeding to the proof. By taking regular neighborhoods of complexes in a triangulation of $\Rspace^4$ one shows

\begin{corollary}\label{cor:manifold}
For any given presentation $P$, there exists a $4$-manifold-with-boundary $M \subset \Rspace^4$ with $\pi_1(M) \cong G(P)$.
\end{corollary}

Next we make the model complexes more economical in terms of the number of simplices in the complex. We define the \emph{binary size} of a presentation as follows. Let $m_g$ be the largest exponent of the generator $g$ or that of its inverse $g^{-1}$ appearing in the words $r_j$. Moreover, denote by $m$ the largest value of the $m_g$. First define the binary size of a word $r_j = {g'_1}^{k_{j1}} \dots {g'_{m_j}}^{k_{jm_j}}$ to be $b(r_j)= \sum_{l=1}^{m_j} {\lceil \log{m_{g'_l}} \rceil}$, where the $g'_l$ are generators or inverse of a generator, $g'_l \neq g'_{l+1}$, and the logarithm is in base 2. The binary size of the presentation is $b(P) = n + \sum_{j=1}^{m} b(r_j)$. We call a model complex \textit{small}, if its size is in the order of the bit-complexity of the presentation, up to logarithmic factors.

\begin{theorem} \label{theorem:bit}
Any given presentation $P$ can be transformed to an equivalent presentation $P'$ whose unary size is $O(b(P) \log{n})$.
\end{theorem}
\begin{proof}
 For each generator $g$ we define $\lceil \log{m_g} \rceil$ new generators called $g^{(1)}, \ldots, g^{(\lceil \log{m_g} \rceil)}$ and relations $g^{(k)} = g^{(k-1)}g^{(k-1)}, k=2,..., \lceil \log{m_g} \rceil, g^{(1)} = g$. For each word $r_i$ we make the obvious replacements, and write it with respect to the new generators. For each generator $g$ (or inverse of a generator) we have added $\lceil \log{m_g} \rceil$ new generators, and we have replaced $g'^{k_{jl}}$ in any word by at most $\lceil \log m_g \rceil$ generators (or inverse of one) each of exponent 1. Each old relation $r_i$ now can be written with $b(r_i)$ many generators or inverse of one each with exponent 1, hence has unary size $b(r_i)$. The new presentation has at most $n \log{m}$ generators.
The sum of the unary sizes of the new relations is at most $3 \sum_g \lceil \log{m_g} \rceil + 2n$ which is $O(b(P) + n)$. Hence, the unary size of the new presentation is $O(n \log{m} + b(P))$.
\end{proof}

\paragraph*{Remark}
The above proof has a geometric side. In the small model simplicial complex, towers of Moebius bands are used to efficiently generate powers of a generator of the fundamental group,
see for instance \cite{FrKr14} for more information about these spaces. A tower of Moebius bands is constructed as follows. Take an arbitrary Moebius band $M$. Let $c_0$ be its
core circle, and $c_1$ be its boundary, their free homotopy classes satisfy $[c_1] = 2 [c_0]$. Take another Moebius band $M'$ with $[c'_1] = 2 [c'_0]$. Glue the core $c'_0$ to the boundary
$c_1$. In the resulting space $[c'_1] = 2[c_0]=2[c'_1]=4[c'_0]$. Continuing this process one obtains a tower of Moebius bands whose last boundary is a power of the first core. In the
small model complex, ``pinched'' images of such towers exist for rapidly creating powers of a generator.

The procedure in the above theorem and the procedure for defining the model complex together define an standard way of constructing a small model complex for any presentation $P$. Hence,
\begin{corollary}
There exist small model complexes for presentations. Moreover, they are realized in $\Rspace^4$, and whose realization can be constructed in almost linear time in bit-complexity of $P$.
\end{corollary}

\paragraph*{}
We also state the following generalization of Theorem \ref{theorem:main} whose proof will be given in the next section.
\begin{theorem} \label{theorem:htype}
For each homotopy type of (finite) $2$-dimensional CW complexes there exists a representative realized in $\Rspace^4$.
\end{theorem}

Any homotopy type of $2$-dimensional CW complexes has a simplicial complex representative, see e.g. \cite{Si93}. It then follows that there exists the corresponding $4$-manifold-with-boundary in $\Rspace^4$ homotopy equivalent to any given $2$-dimensional CW complex.

As the final application, we extend a result about undecidability of relative embedding problems, proved in \cite{FrKr14}, to include the $4$-dimensional case. We are interested to decide for a pair of simplicial complexes $(K,L)$, $L \subset K$, and a fixed PL embedding $f: |L| \rightarrow \Rspace^4$ whether the map $f$ can be extended to an embedding of the whole complex $K$ in $\Rspace^4$ or not. The same problem for $\Rspace^d, d \geq 5$ is undecidable for codimension at most $2$ as proved in \cite{FrKr14}. Here we extend this to the case $\Rspace^4$. Define \emph{refinement complexity} of a PL embedding $f:|K| \rightarrow \Rspace^d$ to be the minimum number of simplices in a subdivision of $K$ on which $f$ is simplex-wise linear.

\begin{theorem} \label{theorem:relembedding}
For each $d \geq 4$ there exists a sequence of $(d-2)$-complex pairs
$(K_i,L_i)$ with a fixed PL embedding $f_i: |L_i| \rightarrow \Rspace^d$ such that the refinement complexity of any extension to an embedding of $K_i$
is a non-recursive function of the number of simplices of $K_i$.
\end{theorem}

The proof of the above theorem is also given in the next section.

\section{Proofs}\label{section:proofs}

We first give the proof of Theorem \ref{theorem:main}. As a first step, we change the presentation $P$ into a new presentation
$$P' = \langle g_1, \ldots, g_n, h_1, \ldots, h_m ; r_1h_1, \ldots, r_mh_m, h_1, \ldots, h_m \rangle.$$

It is easily seen that $G(P') \cong G(P)$. The property of $P'$ which we need is that, for each relation there is a generator appearing with exponent $1$, namely $h_i$. Thus in the following we work with the presentation $P$ which is of this form.

We first take as the wedge of circles $W$ a graph as in Figure \ref{fig:wedge}, and realize $W$ in a $2$-plane $\Pi \subset \Rspace^4$ as is shown in the figure. From now on, by $|W|$ we mean the image of this embedding. In general, we denote the image of an abstract space $X$ in $\Rspace^4$ by $|X|$ to distinguish it from the abstract space. Note that in the embedding of the graph as drawn in Figure \ref{fig:wedge}, any circle $h_j$ contains an edge $e_j$ such that any point in $|e_j|$ has distance at least $\delta > 0$ to $|W| - \bigcup_j |h_j|$. This $\delta$ is a constant of the algorithm and is independent of the given presentation.

\begin{figure}[h]
  \centering
  \includegraphics[scale=0.8]{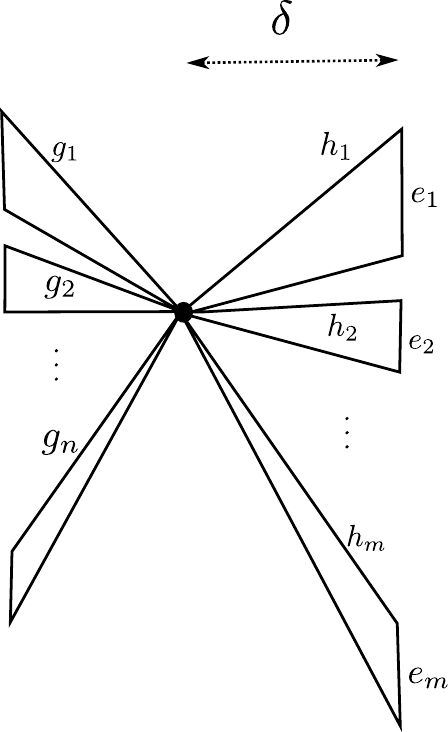}\\
  \caption{the wedge of the generator circles}\label{fig:wedge}
\end{figure}

For each $j=1,\ldots, m$, we have to find a disk $D_j \subset \Rspace^4$, such that its boundary is glued to $|W|$ by a map representing $r_jh_j$ or $h_j \in \pi_1(W)$, and its relative interior is disjoint from the rest of the complex. For the relations of the type $h_j$ we use the region of $\Pi$ that the circle $h_j$ bounds as its disk. Hence, we focus on the relations of the type $r_j h_j$ in the sequel.

We start with a procedure for mapping the annuli $A_j$ defined in Section \ref{section:theorems} into $\Rspace^4$. Observe that $\Rspace^4$ is the union of $3$-planes ``passing through" $\Pi$, where $\Pi$ is their only intersection. These $3$-planes could be parameterized by an angle $\gamma \in [0,\pi)$. We denote by $\Sigma_\gamma$ the 3-plane with parameter $\gamma \in [0, \pi)$. Let $\Lambda \subset \Rspace^4$ be a $3$-plane parallel to $\Pi$, that is $\Lambda \cap \Pi = \emptyset$. Let $\Sigma_0$ be the $3$-plane passing through $\Pi$ and parallel to $\Lambda$. The value $\gamma$ is then the angle between the normal of $\Pi$ that point towards $\Lambda$ in $\Sigma_\gamma$ and a fixed normal to $\Pi$ in $\Sigma_0$.

\paragraph{Remark}
A smooth version of the mapping of the annulus can be thought of as follows. One cuts the annulus into a strip and them glues the lower
boundary of the strip along the required loop on $\Pi$ and rotates the strip such that in each $\Sigma_\gamma$ the strip is a segment perpendicular to $\Pi$. The following
is a PL approximation of this map.

We say an annulus $A$ is an \emph{$n$-annulus} if it is comprised of $n$ rectangles. We denote vertices in one boundary of $A$ by $v_i, i=1, \ldots, n$, and those in the other boundary by $w_i, i=1, \ldots, n$. The boundary circle with the $v_i$ is called the \emph{upper boundary} and the other boundary circle is called the \emph{lower boundary} of $A$, see Figure \ref{fig:annulus}. The annulus is triangulated as shown in the figure. The lower boundary is the circle whose image will be glued to $|W|$.

\begin{figure}[h]
  \centering
  \includegraphics[scale=0.5]{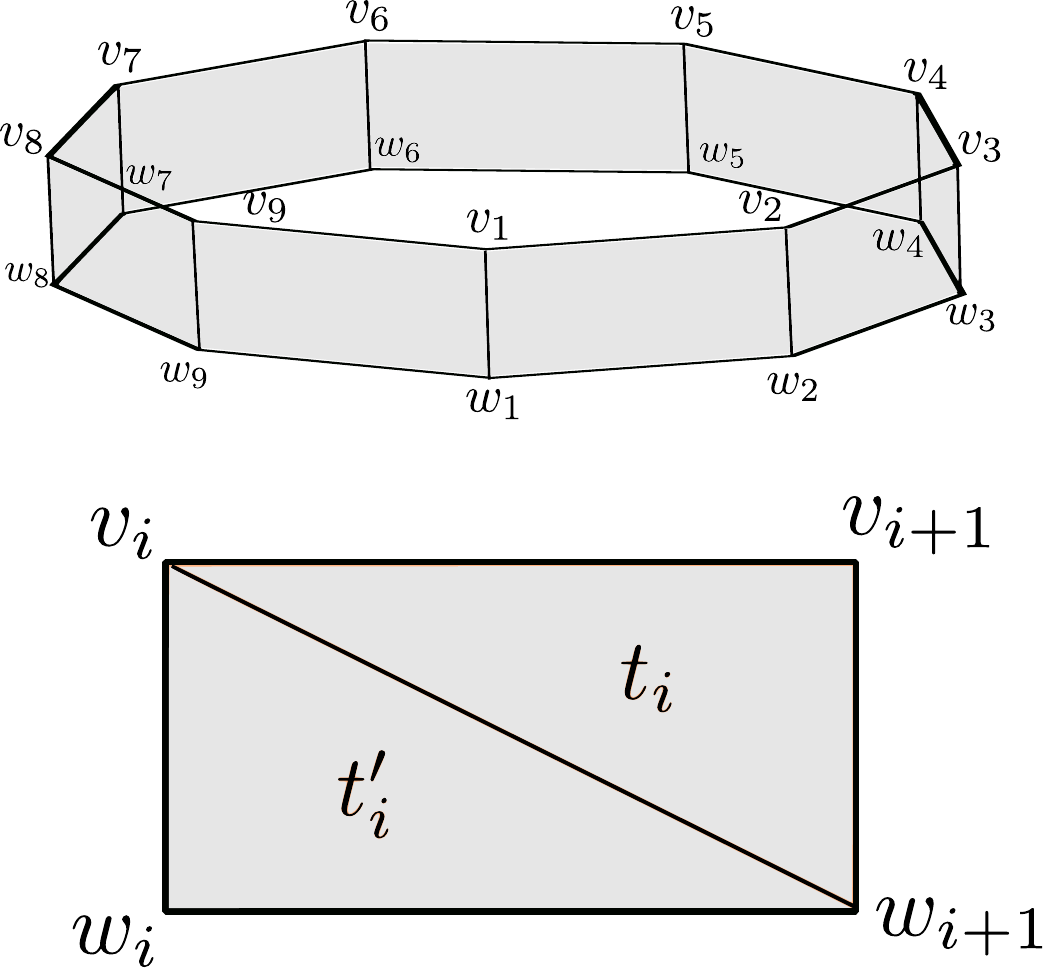}\\
  \caption{notation and triangulation of a $9$-annulus, indexes modulo $n$}\label{fig:annulus}
\end{figure}

For $j=1,\ldots,m$, we take an interval $I_j \subset (\theta^*, \pi/2)$ such that for distinct $j$ the intervals are disjoint. Here $\theta^* < \pi/2$ is an absolute constant.  The lengths of these intervals satisfy $|I_j| < \sigma$. The constant $\sigma$ will be determined later.

Let $I=I_j$. We map a $ck_j$-annulus (see Section 1 for the definition of $k_j$) $A = A_j$ in $\Sigma_{I}=\bigcup_{\gamma \in I} \Sigma_\gamma$, such that the image of its lower boundary is attached to $|W| \subset \Pi$ by the required map representing $r_jh_j \in \pi_1(W)$, and, is embedded elsewhere. Set $n = ck_j$. We show that independent of $n$, such an annulus can always be realized. Let $Q_1, \ldots, Q_n$ be the rectangles of $A$. Each rectangle is divided into two triangles, denoted $t_l, {t_l'}, l=1, \ldots, n$. We consider the obvious attaching map that maps a lower edge of $A$ into one edge of $W$, and which represents $r_jh_j \in \pi_1(W)$. We then identify the vertices (and edges) of the lower boundary with those of $W$. The basepoints are the wedge point in $W$, and $w_2$ in the lower boundary. Note that each lower vertex $w_l$ determines a vertex of $W$. Moreover, there exists exactly one rectangle such that the lower edge of it is identified with $e_j$ by the attaching map and that is $Q_n$. Let $L$ be the smallest subcomplex containing $A- Q_n$. The complex 
$L$ is just a \emph{strip} of $n-1$ rectangles.

Assume $ \theta_1, \theta_2, \ldots, \theta_n \in I$ are $n$ values such that $\theta_1 < \theta_2 < \dotsb < \theta_n$. For $l = 1, \ldots, n$, let $v_l$ be the point where the line perpendicular to $\Pi$ and passing through $w_l \in |W|$ in $\Sigma_{\theta_l}$ intersects $\Lambda$. We have now placed the vertices of $L$ in $\Rspace^4$. If we just realize $L$ linearly using these vertices, we have attached the strip $L$ such that it is disjoint from the rest of the complex other than for its lower boundary which is in $\Pi$. Note that this argument is correct regardless of how the lower edges of the rectangles are identified with segments in $\Pi$. Moreover, observe that the upper boundary of the strip lies entirely in the $3$-plane $\Lambda$.

In the next step, we have to embed the rectangle $Q_n$ to finish constructing the annulus $|A|$. The vertices of the rectangle $Q_n$ already exist, and we can just linearly realize $Q_n$ as two triangles. However, there is no guarantee that the two triangles are disjoint from the rest of the annulus. Nonetheless, we claim that for $\sigma$ sufficiently small the
annulus will be embedded. The intuition behind this is that the intersection of the annulus with $\Sigma_\gamma$, $\gamma \in (\theta_1, \theta_n)$, $\gamma \neq \theta_i$, ought to consist of two almost perpendicular (on $\Pi$) segments whose base points have a positive distance from each other.  Let $|A_j|$ be the annulus in $\Rspace^4$ (possibly with self-intersections) constructed by the above procedure for the $j$-th relation.

\begin{lemma}
The constant $\sigma$ can be chosen such that, for all $j$, the relative interior of the annulus $|A_j|$ is embedded in $\Sigma_{I_j}$, the lower boundary is attached to $|W|$ by a map representing $r_jh_j$, and the upper boundary is a circle in $\Lambda \cap \Sigma_{I_j}$.
\end{lemma}

\begin{proof}

We use the same notation as above. Let $|A| = |A_j|$ and consider a $3$-plane $\Sigma_\gamma, \gamma \in [\theta_1, \theta_n]$, and the set $A_\gamma = Cl(\Sigma_\gamma \cap |A| - \Pi)$. First, observe that $A_{\theta_1}$ and $A_{\theta_n}$ consist of a triangle each, so there is no intersection in these two $3$-planes. Notice that the triangle $|t_n|$ spans all of the $\Sigma_{(\theta_1, \theta_n)}$ and contributes a line segment to $A_\gamma$ with $w_1$ as one endpoint. If $\gamma  = \theta_i$ for some $1<i<n$, then, $A_\gamma$ consists of a triangle perpendicular to $\Pi$ in $\Sigma_\gamma$, namely the triangle $|t'_i|=v_iw_iw_{i+1}$, and, a line segment from the triangle $|t_n|$, whose endpoints are $w_1$ and a point in the segment $v_1v_n \subset \Lambda$, let this line segment be $f_\gamma$. For other values of $\gamma$, $A_\gamma$ consists of two line segments, both having an endpoint in $\Lambda$, see Figure \ref{fig:attaching}. The other endpoints are $w_{i+1} \neq w_1$ and $w_1$, if $\theta_{i} \leq \gamma \leq \theta_{i+1}, i\in \{ 1,\ldots,n-1 \}$. Let $f_\gamma$ denote the segment with endpoint $w_1$ and $e_\gamma$ the other segment. The point $w_1$ is never an endpoint of a $e_\gamma$ since the edge $|e_j|$ is covered once by the attaching map. The figure shows the case $i=n-1$.

\begin{figure}
  \centering
  \includegraphics[scale=0.7]{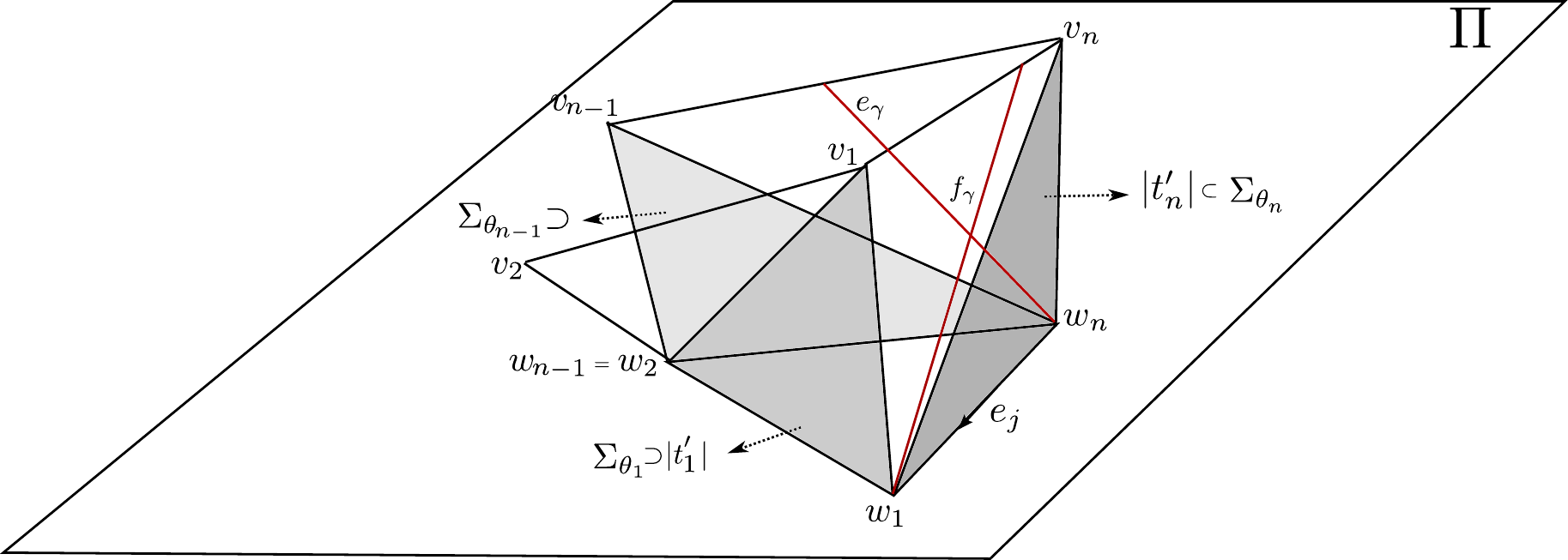}\\
  \caption{part of an embedded annulus}\label{fig:attaching}
\end{figure}

Consider the second case, i.e., when $A_\gamma$ consists of two line segments. We are interested in the angle that $e_\gamma$ and $f_\gamma$ make with the normal of $\Pi$ in $\Sigma_\gamma$ pointing towards $\Lambda$. We can say that the angle between two directions in $\Rspace^4$ is the length of the geodesic that connects the corresponding points in $S^3 \subset \Rspace^4$.

Let $\alpha$ be this angle for $e_\gamma$, and $\beta$ the corresponding one for $f_\gamma$. The angle $\alpha$ is bounded from above by the angle $2\zeta_i = 2(\theta_{i+1} - \theta_i)$. To see this note that $e_{\gamma}$ lies in the triangle $|t_i|$ and is incident on $w_{i+1}$ and thus its angle with $w_{i+1}v_{i+1}$ is bounded by $\zeta_i$. Let $n_{i+1}$ be the normal vector in $\Sigma_{\theta_{i+1}}$ pointing towards $\Lambda$. Then the angle between $n_{i+1}$ and $e_{\gamma}$ is bounded from above by $\zeta_i$. Moreover, if $n_\gamma$ denotes the normal in $\Sigma_{\gamma}$ pointing towards $\Lambda$, then the angle between $n_\gamma$ and $n_{i+1}$ is also bounded by $\zeta_i$. It follows from the triangle inequality on $S^3$ that the angle between $n_{\gamma}$ and $e_\gamma$ is bounded by $2\zeta_i$.

We have thus $\alpha \leq 2\zeta_i < 2|I_j|$. Also, $f_\gamma$ lies in the triangle $|t_n|$ and similar to the above it follows that

$$\beta \leq  2(\theta_n - \theta_1) \leq 2|I_j|. $$

On the other hand, note that one endpoint of $f_\gamma$ is always $w_1$
and the nearest that the endpoint of $e_\gamma$ can come to it is $\delta$ or the length of $e_j$. We can assume that the length of $e_j$ is a constant of the algorithm and is greater than $\delta$.
It follows from Lemma \ref{lemma:segments} below that if $$ \alpha, \beta < \frac{\delta} {2 \rho}$$ an intersection between $e_\gamma$ and $f_\gamma$ is impossible, where $\rho$ is the maximum length of a segment of type $e_\gamma$ or $f_\gamma$. Thus by choosing $\sigma$ such that
$$ 2|I_j| < \frac{\delta}{2 \rho}$$
there will be no intersection.

It is easy to see that the length of the segments $e_\gamma$ and $f_\gamma$ are decreasing functions of $\gamma$ for $\gamma \in (\theta_1, \theta_n)$. Hence, the maximum length is achieved by the length of the ``first" segment of the strip, i.e., $w_1v_1$, and this length is always less than the distance between parallel (in $\Sigma_{\theta^*}$) $2$-planes $\Pi$ and $\Lambda \cap \Sigma_{\theta^*}$. Let $\rho'$ denote this distance. Then $\rho < \rho'$, hence it is enough to choose $\sigma$ such that $$ \sigma  < \frac{\delta} {4\rho'}.$$

A similar argument works for the remaining case.
\end{proof}

\begin{lemma}\label{lemma:segments}
Let $\Pi$ be a $2$-plane in a $3$-space and let $e_1,e_2$ be two segments each having exactly one endpoint in $\Pi$, and so that they intersect at the point $p$ outside of $\Pi$. Moreover, let $\delta$ be the distance between their endpoints in $\Pi$. Then at least one of $e_1,e_2$ makes an angle with the normal vector of $\Pi$ (pointing towards the half-space containing the $e_i$) which is greater than $\frac{\delta}  {2 \ell(e_i)}$, where $\ell(\cdot)$ denotes length.
\end{lemma}

\begin{proof}
Let $\alpha$ be the angle that $e_1$ makes with the normal and $\beta$ the same for $e_2$. Let $x_1,x_2$ be the endpoints of the segments in $\Pi$ respectively. We can assume that the intersection point is the other endpoint of the two segments. Let $q_\alpha$ be a segment of length equal
to $\ell(e_1)$ and perpendicular to $\Pi$ at $x_1$. Similarly, let $q_\beta$ be the segment of the length $\ell(e_2)$ and perpendicular
to $\Pi$ at $x_2$. Let $y_\alpha$, $y_\beta$ be the other endpoints of $q_\alpha$, $q_\beta$ respectively. Consider the circle with center $x_1$ passing through $p$ and $y_\alpha$ and let $a_\alpha$ be the short arc between $y_\alpha$ and $p$ on this circle. Define $a_\beta$ symmetrically.
Then we have
$$ \ell(e_1) \alpha + \ell(e_2)\beta = \ell(a_\alpha)+\ell(a_\beta) > \ell(p y_\alpha)+\ell(p y_\beta) \geq \ell(y_\alpha y_\beta) \geq \delta.$$
It follows that $\max \{\ell(e_1) \alpha, \ell(e_2) \beta \} > \delta/2$, which implies $ \alpha > \delta/(2 \ell(e_1))$ or $\beta > \delta/(2 \ell(e_2))$.
\end{proof}

To finish the proof of the theorem, it is enough to observe that the circle
in $\Lambda \cap \Sigma_{I_j}$ which is the upper boundary of $A_j$ can be coned from a vertex in the other side of $\Lambda$ and inside $\Sigma_{I_j}$. This cone will be an embedded disk in $\Sigma_{I_j}$, and serves to kill the homotopy class of $r_jh_j$. $\qed$

We remark that the condition $\sigma < \delta/4\rho'$ is independent of the given presentation. Hence the real constraint on $|I_j|$ which makes them arbitrarily small when the number of relations is large, comes from the existence of disjoint intervals $I_j$.

\subsection*{Proof of Theorem \ref{theorem:htype}}
We assume the given CW complex is connected. Since for any $2$-dimensional CW complex there exists a simplicial complex homotopy equivalent to it, it is sufficient to show that for any given simplicial $2$-complex $K'$ there exists a homotopy equivalent complex $K$ realized in $\Rspace^4$. We build $K$ out of $K'$ by contracting a spanning tree of the $1$-skeleton. Then $K$ is a CW complex whose cells are simplices and attaching maps are simplicial. Since the $1$-skeleton of $K$ is a wedge of circles it can be embedded in a $2$-plane $\Pi \subset \Rspace^4$. Let $c_1, \ldots, c_n$ be the $1$-cells of $K$ and $\phi_1, \ldots, \phi_m$ be attaching maps of $2$-cells of $K$. Observe that since the $2$-cells come from triangles of $K$ each attaching map defines an element of the homotopy class of the $1$-skeleton represented by at most three distinct generators (or inverse of a generator). The procedure used in the proof of the main theorem applied to the presentation $P = \langle c_1, \ldots, c_n ; [\phi_1], \
\ldots , [\phi_m] \rangle $ builds a simplicial $2$-complex in $\Rspace^4$ which is homotopy equivalent to $K$. This last claim follows since the complex in $\Rspace^4$ and $K$ have the same $1$-skeleton and homotopic attaching maps for $2$-cells, see for instance Theorem 1.6 in \cite{Si93}.

\subsection*{Proof of Theorem \ref{theorem:relembedding}}
The proof for $d\geq 5$ is given in \cite{FrKr14}. Here we explain and expand that proof and show it applies to the case $d=4$.
We start by constructing the complex $K$ for a presentation $P$ of a group whose word problem is undecidable. This implies that there is a sequence of curves $C_i$ in the $1$-skeleton of $K$ which bound images of disks in $K$, with boundary mapped to $C_i$, but there is no recursive function that bounds the complexity of the image, in terms of number of edges in $C_i$. The complex $K_i$ will be built by adding a new circle $h_i$ to the wedge of circles corresponding to generators of $P$, and a new disk defining the relation $C_i = h_i$, where here $C_i$ means the corresponding word to the curve $C_i$. Thus curves $h_i$ and $C_i$ represent the same class in $\pi_1(K_i)$.

The complex $K_i$ embeds in $\Rspace^4$ by Theorem \ref{theorem:main}. Let $N_i$ be a regular neighborhood of the embedded complex. In case $d\geq 5$ it is easily seen that the $h_i$ bound disks in $N_i$ but in our setting this is not immediate. Anyways, we assume for now that the curves $h_i$ bound disks in $N_i$. Again, the complexity of these disks is not bounded by any recursive function, otherwise we could solve the word problem for the group by an algorithm.

We can assume that $N_i$ is triangulated by a subcomplex of a simplicial complex triangulating a ball in $\Rspace^4$. Let $X_i$ be the complex which is the closure of complement of $N_i$ union with the $1$-skeleton of $N_i$.
Now we build the complex $K'_i$ by gluing a disk along $h_i$ to $X_i$ and we consider the pair $(K'_i,X_i)$ with the given embedding of $X_i$. Then it is easy to see that the embedding problem for these codimension $0$ pairs is undecidable since it is equivalent to the triviality problem for the $h_i$ in $N_i$.

Next we reduce the dimension to $d-2$. This is done by removing $d$-simplices and $d-1$-simplices from the complexes $X_i$ and $K'_i$. We have to show that removing these simplices does not affect the computability of our problem. Indeed, the complementary space of the $(d-2)$-skeleton of $X_i$ in the given embedding has fundamental group $G(P)$ free product with a free group.
To see this consider the complementary space of $X_i$ at the beginning as an open subset of $\Rspace^d$, i.e. interior of $N_i$ (Adding $1$-skeleton of $N_i$ does not change fundamental group of complement of $X$). Now removing $d$-simplices has the effect of adding disjoint open $d$-simplices to this complement. Hence the fundamental group does not change. Removing $(d-1)$-simplices from $X_i$ means adding relatively open $(d-1)$-simplices to the complement. Each such simplex is incident on at most two $d$-simplices and changes the fundamental group if and only if the two had been connected before in the space. In other words, the nerve of the collection of $d$-simplices and $(d-1)$-simplices is a graph hence has free fundamental group. Since free product with a finitely generated free group does not change computability it follows that we can replace $(K'_i,X_i)$ with its $(d-2)$-skeleton and the theorem follows. 

It remains to show that the curve $h_i$ bounds a disk in $N_i$. First observe that in (the embedding of) $K_i$, each edge of $h_i$ is incident on only one triangle, and there is a point on an edge of $h_i$ with a positive distance from the complex minus a neighborhood of that point. This implies that there is a transverse sphere for
the triangle incident on that edge. If $C_i$ bounds an image of a disk in $N_i$, then it also bounds an immersed disk which intersects itself in disjoint points and transversely at those points in $N_i$.

Then $h_i$ will bound an immersed disk which consists of the immersed disk bounded by $C_i$ and image of the disk introducing the relation $C _i = h_i$.
Using some transverse spheres of the part of the disk incident on $h_i$ self-intersections of the disk bounded by $h_i$ can be removed by tubing and hence $h_i$ will bound an embedded disk.
We refer to \cite{FrQu90} Chapter 1 for definition of transverse spheres and how to remove intersections using them.

\section{Acknowledgements}
The author expresses his thanks to Herbert Edelsbrunner for his valuable suggestions regarding this research and to Arnaud de Mesmay for helpful discussions.

\bibliographystyle{plain}

\end{document}